\documentclass[journal,10pt]{IEEEtran}
\usepackage{graphicx}
\usepackage{epstopdf}
\usepackage{multirow}
\usepackage{amsmath}
\usepackage{mathtools}
\usepackage{authblk}
\usepackage{cite}
\usepackage{cleveref}
\usepackage{latexsym}
\usepackage[justification=centering]{caption}
\usepackage{amsthm}
\usepackage{balance}
\usepackage{relsize}
\usepackage[nodisplayskipstretch]{setspace}
\usepackage{color}
\usepackage{amssymb}
\usepackage{eucal}
\usepackage{url}
\usepackage{booktabs,subcaption,amsfonts,dcolumn}

\usepackage[nodisplayskipstretch]{setspace}
\usepackage{color}
 \usepackage{soul}
\newcolumntype{d}[1]{D..{#1}}
\usepackage{lipsum}

\newtheorem{lemma}{Lemma}
\usepackage[nodisplayskipstretch]{setspace}
\usepackage{color}
\definecolor{Mypink}{RGB}{255,0,255}
\definecolor{Myorange}{RGB}{255,102,0}
\definecolor{Mygreen}{RGB}{0,153,0}
\definecolor{Myblue}{RGB}{0,0,255}

\DeclareMathAlphabet\mathbfcal{OMS}{cmsy}{b}{n}
\begin{document}

\title{Performance of Reconfigurable Intelligent Surfaces in the Presence of Generalized Gaussian Noise}

\renewcommand\Authfont{\fontsize{12}{14.4}\selectfont}
\renewcommand\Affilfont{\fontsize{9}{10.8}\itshape}

\author{
Lina~Mohjazi,~\IEEEmembership{Senior~Member,~IEEE,}
Lina~Bariah,~\IEEEmembership{Senior~Member,~IEEE,}
        Sami~Muhaidat,~\IEEEmembership{Senior~Member,~IEEE,}
        and Muhammad Ali Imran,~\IEEEmembership{Senior Member,~IEEE.}
				\thanks{L. Mohjazi and M. A. Imran are with the James Watt School of Engineering, University of Glasgow, Glasgow, G12 8QQ, UK (e-mail: l.mohjazi@ieee.org, $\text{Muhammad.Imran}\text{@glasgow.ac.uk}$).} 
				\thanks{L. Bariah is with the KU Center for Cyber-Physical Systems, Department of Electrical Engineering and Computer Science, Khalifa University, Abu Dhabi 127788, UAE, and with the James Watt School of Engineering, University of Glasgow, Glasgow, G12 8QQ, UK (e-mail: lina.bariah@ieee.org).}
				\thanks{S. Muhaidat is with the KU Center for Cyber-Physical Systems, Department of Electrical Engineering and Computer Science, Khalifa University, Abu Dhabi 127788, UAE, and with the Department of Systems and Computer Engineering, Carleton University, Ottawa, ON K1S 5B6, Canada, (e-mail: muhaidat@ieee.org).}
		\vspace{-0.7cm}		
		}

\maketitle

\begin{abstract}
\boldmath
In this letter, we investigate the performance of reconfigurable intelligent surface (RIS)-assisted communications, under the assumption of generalized Gaussian noise (GGN), over Rayleigh fading
channels. Specifically, we consider an RIS, equipped with $N$ reflecting elements, and derive a novel closed-form expression for the symbol error rate (SER) of arbitrary modulation schemes. The usefulness of the derived new expression is that it can be used to capture the SER performance in the presence of special additive noise distributions such as Gamma, Laplacian, and Gaussian noise. These special cases are also considered and their associated asymptotic SER expressions are derived, and then employed to quantify the achievable diversity order of the system. The theoretical framework is corroborated by numerical results, which reveal that the shaping parameter of the GGN ($\alpha$) has a negligible effect on the diversity order of RIS-assisted systems, particularly for large $\alpha$ values. Accordingly, the maximum achievable diversity order is determined by $N$.

\end{abstract} 


\vspace{-0.2cm}

\IEEEpeerreviewmaketitle
\section{Introduction}

\IEEEPARstart{S}{MART} radio environments, empowered by reconfigurable intelligent surfaces (RISs), are envisioned to revolutionize the design of wireless systems, since they offer the ability to control the propagation environment of electromagnetic waves through establishing desired channel responses. Unlike conventional transmission techniques, such as phased arrays, multi-antenna transmitters, and relays involving active components, RISs are composed of a large number of passive reflecting elements (REs), each is supported by the lowest number of small-sized, low-power, and inexpensive components \cite{DiRenzo2020}. With the aid of at least one smart controller, the amplitude and/or phase-shift of incident signals are manipulated at each RE in real time, allowing for a transformative control (i.e., scattering, reflection, absorption, and refraction) of the radio waves characteristics. 
%
%
%
%
%
%
\par To reap the full potentials brought by RIS and cater for their unique challenges, exploring the performance limits of RIS-assisted wireless communications analytically has become an active research topic \cite{bariah2020,Boulogeorgos,Han,Atapattu}. Specifically, in \cite{bariah2020}, the error probability performance is investigated for RIS-assisted non-orthogonal multiple access (NOMA) and backscatter communication systems, respectively. Furthermore, \cite{Boulogeorgos} investigated the performance of RIS-assisted wireless systems, based on outage probability, symbol error rate, and ergodic capacity. Similarly, the ergodic capacity is studied in \cite{Han} for multiple-antenna RIS-assisted systems. For two–way communications, the outage probability and spectral efficiency of the system are explored in \cite{Atapattu}. Apart from the previous works considering ideal operational conditions, a few research studies examined the performance of RIS-assisted systems in the presence of several practical considerations, such as RF transceiver hardware impairments \cite{Hemanth} and phase noise \cite{qian2020beamforming}.  
\par Although the results from \cite{bariah2020,Boulogeorgos,Han,Atapattu,Hemanth,qian2020beamforming} are insightful, the performance of the considered systems is studied under the additive white Gaussian noise (AWGN) assumption. While this assumption accurately models the distribution of thermal noise, which is due to the random motion of free electrons at the receiver, it ignores the impulsive nature of electromagnetic interference, atmospheric noise and/or man-made noise, which might be dominant in numerous applications and lead to severe performance degradation. In particular, it was shown that communication systems designed under the AWGN assumption typically suffer from severe performance degradation when exposed to impulsive noise \cite{Mohjazi2018}. In this regard, automotive ignition noise and power transmission lines are common sources of impulsive noise, which are experienced, particularly in urban areas \cite{soury2012}. Also, the combined statistics of interference and noise in ultra wide-band communication systems is demonstrated to follow the generalized Gaussian distribution (GGD). It is worth mentioning that the impulsive noise component in the simplified class-A noise model, which is modeled by Laplacian distribution, is a special case of the GGD. Note that the additive white generalized Gaussian noise (AWGGN), given in terms of the generalized Q–function, can model different types of noise, e.g., AWGN, Laplacian, Gamma, uniform, and impulsive noise \cite{9097184}.

\par We emphasize that the impact of AWGGN on the performance of RIS-assisted communications is not understood yet. This underlines the need for carrying a thorough investigation, which is compelling to provide pragmatic information for the system designer and hence, facilitates a successful realization of RIS-assisted systems in specific applications. 
\par Motivated by the above consideration, in this letter, we analyze the performance of an RIS-assisted wireless communication system over Rayleigh fading channels, in the presence of AWGGN. We first derive a novel closed-form expression for the symbol error rate (SER), which is subsequently used to quantify the SER and the achievable diversity order of the systems when three types of noise are considered, namely, Gamma, Laplacian, and Gaussian noise. To the best of the authors’ knowledge, such a performance study has not been addressed in the open literature.
%
%
\section{System Model}\label{model}
In this letter, we consider that a single-antenna source node, $S$, communicates with a single-antenna destination node, $D$, through an RIS, which is deployed to extend the operational range of the system. The end-to-end (E2E) signal-to-noise ratio (SNR) between $S$ and $D$ defines the quality of the overall communication, including the SER. Therefore, to quantify the impact of AWGGN on RIS-assisted networks, in this section, we analytically present the distribution of the instantaneous SNR, which will be exploited next to develop the statistical characterization of the underlying SER and the achievable diversity order of the system in the presence of AWGGN.
\par We assume that the direct link between $S$ and $D$ is subject to strong attenuation, as a result of deep fading, or shadowing effects, caused by surrounding physical obstacles, or both, and therefore, this link does not exist. It is worth mentioning that such an assumption is followed in the literature of RIS-assisted communications \cite{bariah2020,Boulogeorgos}, where RIS is considered to be deployed in environments that experience high blockage.
\par In our setup, we consider an RIS that is equipped with $\textit{N}$ REs, each of which can be reconfigured by a communication oriented software through a controller. The signal transmitted from $S$, being either a base station or an access point, is reflected by the RIS towards $D$. The small-scale complex channel fading coefficients of the $S$ $\to$RIS and RIS$\to D$ links are denoted as $h_i$ and $g_{i}$, respectively, where $i\in\{1,2,...,N\}$ is the $i$-th element of RIS. It is assumed that the fading envelopes of the two wireless links are modeled as independent and identically distributed (i.i.d) Rayleigh fading channels with zero mean and unit variance, i.e., $\sigma^2=1$. The assumption of Rayleigh fading channels presents the scenario in which a line-of-sight (LOS) propagation cannot be established due to random RIS deployments, e.g., if the RISs are deployed on spatial blockages. In such a case, the system designer has no control over optimizing the locations of the RISs \cite{qian2020beamforming}. \footnote{It is worth mentioning that this assumption is followed in previously published works including \cite{Boulogeorgos,Atapattu} and the references therein. This assumption stems from the fact that even if LOS links between S$\to$RIS and RIS$\to$$D$ are not blocked, there still exist a high number of scatters.}
 \par If the source node transmits the symbol $x$ with power $P_s$, then the baseband equivalent received signal at $D$ through RIS is expressed as \vspace{-0.2cm}
 \begin{equation}\label{Pr}
y=\sum_{i=1}^N h_i g_{i}e^{j\theta_i}x+n, 
 \end{equation} 
 where the noise term, $n\in \mathbb{R}$, represents the AWGGN with zero mean and variance $N_0/2$ \cite{soury2012}. Specifically, the probability density function (PDF) of $n$ is given by \cite{soury2012}
\begin{equation}
f(n)=\frac{\alpha \Lambda}{2\Gamma (1/\alpha)}\text{exp}\left(-\Lambda^\alpha|n|^\alpha\right),
\end{equation} 
where $\Gamma(.)$ represents the complete Gamma function \cite{IntTable}, $\alpha \in \mathbb{R}^+$ denotes the shaping parameter, and $\Lambda=2\Lambda_0/N_0$ is the noise power normalization coefficient with $\Lambda_{0} = \sqrt{\Gamma(3/\alpha)/\Gamma(1/\alpha)}$. It is worth mentioning that when the value of $\alpha$ increases, the GGD tail becomes tighter, yielding to lower noise levels. On the other hand, as $\alpha$ decreases, the impact of AWGGN on the received signal becomes higher.

\par Moreover, in \eqref{Pr}, $\theta_i$ specifies the adjustable phase induced by the $i$-th RE of the RIS \cite{Boulogeorgos}. We further consider that optimal phase shifting, i.e., $\theta_i=-(\phi_{h_i}+\phi_{g_{i}})$, is provided by the RIS, since it is assumed that the RIS has perfect knowledge of the channel phases of $h_i$ and $g_{i}$, denoted respectively as $\phi_{h_i}$ and $\phi_{g_{i}}$ \cite{bariah2020,qian2020beamforming,Boulogeorgos}. In this letter, we assume that the amplitude of the reflection coefficient of each RE is equal to 1 for all $i\in N$ \cite{Boulogeorgos}. Capitalizing on this and based on \eqref{Pr}, the instantaneous received SNR at $D$ through the RIS is expressed as
\begin{equation}\label{Pr-D}
 \gamma= \left|\sum_{i=1}^N |h_i| |g_{i}|\right|^2 \bar{\gamma},
\end{equation} where $\bar{\gamma}=2P_s/\left(d_1^{\varrho}d_2^{\varrho} N_0\right )$ specifies the average transmit SNR, with $d_1$, $d_2$, $\varrho$ denoting the $S$-RIS distance, RIS-$D$ distance, and path-loss exponent, respectively. The PDF of the E2E SNR is delivered in the following Lemma.
\newcounter{tempequationcounter}
 \begin{figure*}[ht]
 \setcounter{equation}{4}
\begin{equation}\label{m3a}
\mu_{3}=\left\{ \begin{array}{l}
N\pi\left(9/2+6(N-1)+(N-1)(N-2)\pi^{2}/8\right),N\geq3\\
21\pi,N=2\\
9\pi/2,N=1.
\end{array}\right. 
\end{equation}
\end{figure*}
\begin{figure*}[ht] 
\begin{align}\label{mu4}
\mu_{4}=\left\{ \begin{array}{l}
\left(64N+48N(N-1)+9N(N-1)\pi^{2}+6N(N-1)(N-2)\pi^{2}+N(N-1)(N-2)(N-3)\pi^{4}/16\right),N\geq4\\
\left(480+90\pi^{2}\right),N=3\\
\left(224+18\pi^{2}\right),N=2\\
64,N=1.
\end{array}\right.
\end{align}
\hrulefill
\end{figure*}
 \setcounter{equation}{3}
\begin{lemma}
\label{lemma:SNR_PDF}
The PDF of the SNR, $\gamma$, can be tightly approximated by \cite{mohjazi2020}
\begin{equation}
\label{eq:PDF}
f_\gamma\left ( \gamma \right )= \frac{a_1 a_2}{2 \gamma}\textup{G}^{2,0}_{1,2}\left [ \frac{1}{a_2}\sqrt{\frac{\gamma}{\bar{\gamma}}} \Bigg| \begin{matrix}
-; a_3 +1\\ 
a_5+1, a_4+1;-
\end{matrix} \right ], \gamma \geq 0
\end{equation}
where $\textup{G}^{m,n}_{p,q}\left [ x \Bigg| \begin{matrix}-; -\\ -;-\end{matrix} \right ]$ represents the Meijer's $\textup{G}$-Function \cite{IntTable}. Moreover, $a_{1}=\Gamma(a_{3}+1)/[a_{2}\Gamma(a_{4}+1)\Gamma(a_{5}+1)]$, $a_{2}=a_{3}/2\left(\varphi_{4}-2\varphi_{3}+\varphi_{2}\right)+2\varphi_{4}-3\varphi_{3}+\varphi_{2}$, $a_{3}=[4\varphi_{4}-9\varphi_{3}+6\varphi_{2}-\mu_{1}]/[-\varphi_{4}+3\varphi_{3}-3\varphi_{2}+\mu_{1}]$, $a_{4}=[a_{6}+a_{7}]/2$, $a_{5}=[a_{6}-a_{7}]/2$, $a_{6}=a_{2}^{-1}[a_{3}\left(\varphi_{2}-\mu_{1}\right)+2\varphi_{2}-\mu_{1}]-3$, $a_{7}=\sqrt{\left(a_{2}^{-1}[a_{3}\left(\varphi_{2}-\mu_{1}\right)+2\varphi_{2}-\mu_{1}]-1\right)^{2}-4a_{2}^{-1}\mu_{1}(a_{3}+1)}$,
and $\varphi_{i}=\mu_{i}/\mu_{i-1},i \geq 1$, where $\mu_{i}$ is the $i$th moment of $\sqrt{\gamma}$, $\mu_{0}=1$. For accurately approximating the PDF in \eqref{eq:PDF}, the first four moments are derived in \cite{bariah2020} and are presented here for the clarity and completeness of the work, as follows, $\mu_{1} =N\pi/2$, and $\mu_{2}  =\left(4+(N-1)\pi^{2}/4\right)N$.
$\mu_3$ and $\mu_{4}$ are given in \eqref{m3a} and \eqref{mu4}, respectively, at the top of next page.
\end{lemma}
\begin{proof}
Proof of Lemma \ref{lemma:SNR_PDF} is provided in \cite{mohjazi2020}.
\end{proof}
It is worth mentioning that, although recent works available in the relevant literature, such as those in \cite{Lyu2020, Lyu2021}, adopt the moments method to approximate the PDF of the E2E SNR, our main contribution compared to those studies is the derivation of a novel tractable closed-form expression for the SER over AWGGN. This will enable us to further offer novel simple-to-evaluate closed-form expressions for the SER considering three different cases of AWGGN types, together with their respective high-SNR asymptotic SER closed-forms. We will also derive novel analytical representations of the underlying diversity order of the three cases under study. The afore-mentioned contributions will be detailed in the next section.
\section{Performance Analysis} 
\label{sec:SER}
\par In this section, we quantify the performance of RIS-assisted communications under different noise models by deriving accurate expressions for the SER and the associated achievable diversity orders. Considering the presence of AWGGN, the conditional SER of an arbitrary modulation scheme can be written as \cite{soury2013}
\setcounter{equation}{6}
\begin{equation}
\label{eq:SER1}
P_{e} \mid \gamma = A \textup{Q}_{\alpha}(\sqrt{B \gamma}),
\end{equation} 
where $A$ and $B$ are modulation dependent parameters, as depicted in Table \ref{Table:AB} on the next page. Also, 
\begin{equation}
\textup{Q}_{\alpha}(x) = \frac{\alpha \Lambda_{0}}{2\Gamma(1/\alpha)}\int_{x}^{\infty}\textup{exp}\left ( - \Lambda_{0}^{\alpha}t^{\alpha}\right )dt
\end{equation}
denotes the generalized Gaussian $\textup{Q}$-function. In order to obtain the closed-form expression of the unconditional SER, we average the conditional SER over the PDF of $\gamma$, given in \eqref{eq:PDF}. 
Therefore, using Lemma \ref{lemma:SNR_PDF}, the unconditional SER can be evaluated as
\begin{equation}
\label{eq:SER2}
\begin{split}
P_{e} &= A\int_{0}^{\infty} \textup{Q}_{\alpha}(\sqrt{B \gamma})f_\gamma\left ( \gamma \right )d\gamma\\
& = \frac{Aa_1 a_2}{2} \int_{0}^{\infty} \frac{1}{\gamma}\\
& \;\;\;\; \times \textup{Q}_{\alpha}(\sqrt{B \gamma})\textup{G}^{2,0}_{1,2}\left [ \frac{1}{a_2}\sqrt{\frac{\gamma}{\bar{\gamma}}} \Bigg| \begin{matrix}
-; a_3 +1\\ 
a_5+1, a_4+1;-
\end{matrix} \right ]d\gamma.
\end{split}
\end{equation} 
To evaluate the integral in \eqref{eq:SER2}, we utilize \cite[Eq. A.3]{soury2012} and \cite[Eq. 8.4.16.2]{int3} in order to represent the generalized Gaussian $\textup{Q}$-Function in terms of Meijer's $\textup{G}$-function, as the following
\begin{equation}
    \label{eq:marcumQ}
    \textup{Q}_{\alpha}(\sqrt{B \gamma}) = \frac{\Lambda_{0}^{\frac{2}{\alpha}-1}}{2\Gamma\left ( \frac{1}{\alpha} \right )}\textup{G}^{2,0}_{1,2}\left [ \Lambda_{0}^{\alpha}\left (\sqrt{B \gamma}  \right )^{\alpha} \Bigg| \begin{matrix}
-; 1\\ 
0, \frac{1}{\alpha};-
\end{matrix} \right ].
\end{equation}
Hence, the integral in \eqref{eq:SER2} can be rewritten as
\begin{equation}
\label{eq:SER3}
\begin{split}
P_{e}  = \frac{Aa_1 a_2 \Lambda_{0}^{\frac{2}{\alpha}-1}}{4\Gamma\left ( \frac{1}{\alpha} \right )}&\int_{0}^{\infty} \frac{1}{\gamma}\textup{G}^{2,0}_{1,2}\left [ \Lambda_{0}^{\alpha}\left (\sqrt{B \gamma}  \right )^{\alpha} \Bigg| \begin{matrix}
-; 1\\ 
0, \frac{1}{\alpha};-
\end{matrix} \right ]\\
& \times \textup{G}^{2,0}_{1,2}\left [ \frac{1}{a_2}\sqrt{\frac{\gamma}{\bar{\gamma}}} \Bigg| \begin{matrix}
-; a_3 +1\\ 
a_5+1, a_4+1;-
\end{matrix} \right ]d\gamma.
\end{split}
\end{equation} 
Exploiting the integration by substitution method with $x=\sqrt{\gamma}$, the integral in \eqref{eq:SER3} can be expressed as
\begin{equation}
\label{eq:SER4}
\begin{split}
P_{e}  = \frac{Aa_1 a_2 \Lambda_{0}^{\frac{2}{\alpha}-1}}{2\Gamma\left ( \frac{1}{\alpha} \right )}&\int_{0}^{\infty} \frac{1}{x}\textup{G}^{2,0}_{1,2}\left [ \Lambda_{0}^{\alpha}B^{\alpha/2}x^\alpha \Bigg| \begin{matrix}
-; 1\\ 
0, \frac{1}{\alpha};-
\end{matrix} \right ]\\
& \times \textup{G}^{2,0}_{1,2}\left [ \frac{1}{a_2}\frac{x}{\sqrt{\bar{\gamma}}} \Bigg| \begin{matrix}
-; a_3 +1\\ 
a_5+1, a_4+1;-
\end{matrix} \right ]dx.
\end{split}
\end{equation} 
Next, the closed-form expression of the SER can be obtained from \eqref{eq:SER4} by employing \cite[Eq. 2.24.1.1]{int3}. Therefore, the closed-form SER can be evaluated to \eqref{eq:SER5}, presented at the top of the next page, where $k$ and $l$ are selected to be integer values and to satisfy $l/k=\alpha$. It is important to note that this closed-form result \eqref{eq:SER5} is a novel generic expression for the SER of arbitrary signaling in RIS-assisted systems subject to AWGGN. It is recalled that most of the common noise models are considered as special cases of the AWGGN. To demonstrate the usefulness of the derived expression in \eqref{eq:SER5} and to obtain more insights into the system performance, as well as its achievable diversity order, we consider the following cases of well-known noise models:  i) Gamma, ii) Laplacian, and iii) Gaussian.
\begin{table}[h]
\centering
\caption{$A$ and $B$ values \cite{salahat2014}.}
\begin{tabular}{|c|c|c|}
\hline
Modulation & $A$ & $B$ \\ \hline\hline
 BPSK   & 1 & 1 \\ \hline
 QPSK, 4-QAM   & 2 & 2 \\ \hline
 $M$-PSK   & 2 & $2\textup{sin}^{2}(\pi/M)$ \\ \hline
Rectangular $M$-QAM   & $4(\sqrt{M}-1)/\sqrt{M}$ & $3/(\sqrt{M}-1)$ \\ \hline
\end{tabular}
\label{Table:AB}
\end{table}
\begin{figure*}[!t]
\begin{equation}
\label{eq:SER5}
\begin{split}
P_{e}  = \frac{Aa_1 a_2 \Lambda_{0}^{\frac{2}{\alpha}-1}k^{(\frac{1}{\alpha}-\frac{1}{2})}l^{(a_5+a_4-a_3+\frac{1}{2})}}{2\Gamma\left ( \frac{1}{\alpha} \right )\left(2 \pi \right)^{\frac{1}{2}(l-1)+\frac{1}{2}(k-1)}} \textup{G}^{2k,2l}_{k+2l,2k+l}\left [ \frac{l^l}{k^k}\Lambda_0^{\alpha k}B^{\frac{\alpha k}{2}}a_{2}^{l}\bar{\gamma}^{\frac{l}{2}} \Bigg| \begin{matrix}
\frac{-a_5}{l}\cdots \frac{-a_5+l-1}{l},\frac{-a_4}{l}\cdots \frac{-a_4+l-1}{l},\frac{1}{k}\cdots \frac{1+k-1}{k}\\ 
0 \cdots \frac{k-1}{k},\frac{1/\alpha}{k}\cdots \frac{1/\alpha+k-1}{k},\frac{-a_3}{l}\cdots\frac{-a_3+l-1}{l}
\end{matrix} \right ].
\end{split}
\end{equation} 
\hrulefill
\vspace*{1pt}
\end{figure*}
%
%
%
%
%
\par \textbf{Case 1: (Gamma noise)}
For the special case of Gamma noise, the shaping parameter $\alpha$ is set to 1/2, i.e., $l=1$ and $k=2$. Therefore, the SER in the presence of Gamma noise can be written as 
\begin{equation}
\label{eq:Gamma}
\begin{split}
P_{e,\textup{GM}}  = \frac{Aa_1 a_2 (5!)^\frac{3}{2}}{\sqrt{\pi} } \textup{G}^{4,2}_{4,5}\left [ \zeta \Bigg| \begin{matrix}
-a_5,-a_4; \frac{1}{2},1\\ 
0,\frac{1}{2},1,\frac{3}{2};-a_3
\end{matrix} \right ]
\end{split}
\end{equation}
where $\zeta = a_{2}\sqrt{B\bar{\gamma}(5!)}/4.$
Utilizing the accurate Meijer’s G-function approximation in \cite[Eq. 41]{ansari2015}, the SER expression in \eqref{eq:Gamma} can be asymptotically evaluated, when $\bar{\gamma} \rightarrow \infty$, as
\begin{equation}
\label{eq:Gamma1}
\begin{split}
\bar{P}_{e,\textup{GM}}  \approx & \frac{Aa_1 a_2 (5!)^\frac{3}{2}}{\sqrt{\pi} }\Bigg [ \tfrac{\Gamma(-a_5+a_4)\Gamma(1+a_5)\Gamma(\frac{5}{2}+a_5)}{\Gamma(-a_5+a_3)}\zeta^{-a_5-1}\\
& \;\;\;\;\;\;\;\;\;\;\;\;\;\;\;\;\; + \tfrac{\Gamma(-a_4+a_5)\Gamma(1+a_4)\Gamma(\frac{5}{2}+a_4)}{\Gamma(-a_4+a_3)}\zeta^{-a_4-1} \Bigg ].
\end{split}
\end{equation}
The asymptotic SER derived in \eqref{eq:Gamma1} can be employed to compute the achievable diversity order of the adopted system model under the assumption of Gamma noise. The achievable diversity order can be defined as the slop of the SER at high SNR values, which can be analytically evaluated as 
\begin{equation}
    \label{eq:Div_Gamma}
    D = -\textup{log}\bar{P}_e/\textup{log} \bar{\gamma}.
\end{equation}
Therefore, to obtain the achievable diversity order, we insert \eqref{eq:Gamma1} into \eqref{eq:Div_Gamma} and perform the log operation. Then by considering the dominant components that are dependent on the SNR, the achievable diversity order can be evaluated as
\begin{equation}
    \label{eq:Div_Gamma2}
    \begin{split}
    D_\textup{GM} & \approx -\frac{\textup{log}\left (\bar{\gamma}^{-\frac{1}{2}(a_5+1)}+\bar{\gamma}^{-\frac{1}{2}(a_4+1)}\right )}{\textup{log} \bar{\gamma}} \\
    & \approx \textup{min}\left [ \frac{1}{2}(a_5+1),\frac{1}{2}(a_4+1)\right ].
    \end{split}
\end{equation}
It can be noted from Lemma 1 that the PDF parameters, including $a_4$ and $a_5$ are functions of $N$. In this respect, it was shown in \cite[Fig. 2]{bariah2020} that $a_5 \leq a_4$, for all $N$ values. Hence, the achievable diversity order expression in \eqref{eq:Div_Gamma2} can be further simplified to $D_\textup{GM}  \approx (a_5+1)/2$.
%
%
%
\par \textbf{Case 2: (Laplacian noise)} In the Laplacian noise model, the noise shaping parameter equals to $\alpha=1$, i.e., $l=1$ and $k=1$. Therefore, for the Laplacian noise scenario, the SER can be obtained by substituting $\alpha$, $l$, and $k$ by their values in \eqref{eq:SER5}, yielding a simplified SER expression, as follows
\begin{equation}
\label{eq:Laplacian}
\begin{split}
P_{e,\textup{LP}}  = \frac{Aa_1 a_2}{\sqrt{2}} \textup{G}^{2,2}_{3,3}\left [ \lambda \Bigg| \begin{matrix}
-a_5,-a_4; 1\\ 
0,1;-a_3
\end{matrix} \right ],
\end{split}
\end{equation}
where $\lambda = a_2 \sqrt{2B\bar{\gamma}}$.
The achievable diversity order under the assumption of Laplacian noise can be evaluated by first deriving the asymptotic SER, which can be evaluated at high SNR values as follows 
\begin{equation}
\label{eq:Laplacian1}
\begin{split}
\bar{P}_{e,\textup{LP}}  \approx  \frac{Aa_1 a_2}{\sqrt{2}}& \Bigg [ \frac{\Gamma(-a_5+a_4)\Gamma(1+a_5)}{\Gamma(-a_5+a_3)}\lambda^{-a_5-1}\\
& + \frac{\Gamma(-a_4+a_5)\Gamma(1+a_4)}{\Gamma(-a_4+a_3)}\lambda^{-a_4-1} \Bigg ].
\end{split}
\end{equation}
Following the same procedure used to obtain the diversity order in Case 1, the achievable diversity order for the case of Laplacian noise can be written as $D_\textup{LP}  \approx (a_5+1)/2$.
\par \textbf{Case 3: (Gaussian noise)} To derive the SER in the presence of Gaussian noise, we assign the values 2, 2, and 1 to the variables $\alpha$, $l$, and $k$, respectively, in \eqref{eq:SER5}, yielding
\begin{equation}
\label{eq:Gauss}
\begin{split}
P_{e,\textup{GS}}  = &\tfrac{Aa_1 a_2}{2^{a_5 + a_4 - a_3 -1}\pi^{3/4}}  \textup{G}^{2,4}_{5,4}\left [ \eta \Bigg| \begin{matrix}
\frac{-a_5}{2},\frac{-a_5+1}{2},\frac{-a_4}{2},\frac{-a_4+1}{2}; 1\\ 
0,\frac{1}{2};\frac{-a_3}{2},\frac{-a_3+1}{2}
\end{matrix} \right ],
\end{split}
\end{equation}
where $\eta = 2a_{2}^{2}B\bar{\gamma}$.
Similar to \eqref{eq:Gamma1} and \eqref{eq:Laplacian1}, the asymptotic SER in the presence of Gaussian noise is given by
\begin{equation}
\label{eq:Gauss1}
\begin{split}
&\bar{P}_{e,\textup{GS}}  \approx  \frac{2^{a_5 + a_4 - a_3 -1}Aa_1 a_2}{\pi^{3/4}} \\
&\Bigg [ \frac{\Gamma(-\frac{1}{2})\Gamma(\frac{-a_5+a_4}{2})\Gamma(\frac{-a_5+a_4-1}{2})\Gamma(1+\frac{a_5}{2})\Gamma(\frac{a_5+3}{2})}{\Gamma(\frac{a_5+4}{2})\Gamma(\frac{-a_5+a_3}{2})\Gamma(\frac{-a_5+a_3-1}{2})}\eta^{-\frac{a_5}{2}-1}\\
& + \frac{\Gamma(\frac{1}{2})\Gamma(\frac{-a_5+a_4+1}{2})\Gamma(\frac{-a_5+a_4}{2})\Gamma(1+\frac{a_5}{2})\Gamma(\frac{a_5-1}{2})}{\Gamma(\frac{a_5+3}{2})\Gamma(\frac{-a_5+a_3}{2})\Gamma(\frac{-a_5+a_3+1}{2})}\eta^{-\frac{a_5}{2}-\frac{1}{2}} \\
& +\frac{\Gamma(-\frac{1}{2})\Gamma(\frac{-a_4+a_5}{2})\Gamma(\frac{-a_4+a_5-1}{2})\Gamma(1+\frac{a_4}{2})\Gamma(\frac{a_4+3}{2})}{\Gamma(\frac{a_4+4}{2})\Gamma(\frac{-a_4+a_3}{2})\Gamma(\frac{-a_4+a_3-1}{2})}\eta^{-\frac{a_4}{2}-1} \\
& + \frac{\Gamma(\frac{1}{2})\Gamma(\frac{-a_4+a_5+1}{2})\Gamma(\frac{-a_4+a_5}{2})\Gamma(1+\frac{a_4}{2})\Gamma(\frac{a_4-1}{2})}{\Gamma(\frac{a_4+3}{2})\Gamma(\frac{-a_4+a_3}{2})\Gamma(\frac{-a_4+a_3+1}{2})}\eta^{-\frac{a_4}{2}-\frac{1}{2}} \Bigg ].
\end{split}
\end{equation}
The achievable diversity order for the case of Gaussian noise can be evaluated using \eqref{eq:Div_Gamma} and \eqref{eq:Gauss1} to $D_\textup{GS}  \approx (a_5+1)/2$.
\begin{figure}[!t]
\centering
\includegraphics[width=0.8\linewidth]{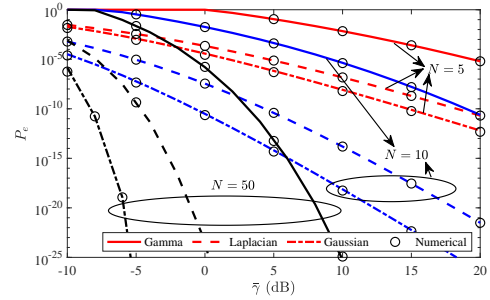}
\caption{Average SER versus $\bar{\gamma}$ in the presence of Gamma, Laplacian, and Gaussian noise, $N=5,10$ and 50.}
\label{fig:Pe_SNR}
\vspace*{-0.5cm}
\end{figure}
\underline{\textit{Remark:}} The results above suggest that the achievable diversity of the system is independent from the type of the noise and is only affected by the number of RIS elements, $N$. This will be also verified through the numerical results.
\section{Numerical Results} 
\label{sec:result}

\par In this section, numerical results are presented to validate the accuracy of the proposed theoretical framework, and to examine the performance of RIS-assisted communications subject to AWGGN over Rayleigh fading, when QPSK signaling is considered, i.e., $A=B=2$. The SER of other modulation schemes can be studied in a similar manner by selecting different values of $A$ and $B$ according to Table \ref{Table:AB}. Unless mentioned otherwise, $d_1$ and $d_2$ are assumed to be normalized to unity.
\par In Fig. \ref{fig:Pe_SNR}, we investigate the SER performance as a function of the average transmit SNR, $\bar{\gamma}$, when the received signals are corrupted by different types of noise, namely, Gamma, Laplacian, and Gaussian noise, and when $N=5,10$ and 50. Note that the generalized $Q$-function is not available in the Matlab library, and under the considered system parameters, other representations of the generalized $Q$-function (including the incomplete gamma function) cannot be accurately evaluated due to the precision limitations in Matlab. Hence, in this work, we selected the numerical approach, which is widely accepted in the relevant literature \cite{soury2012}, to prove the accuracy of our framework. In Fig. \ref{fig:Pe_SNR}, the markers denote numerical results while solid and dashed lines represent analytical results, which are obtained using the closed-form expressions given in \eqref{eq:Gamma}, \eqref{eq:Laplacian}, and \eqref{eq:Gauss}. The full agreement between the numerical and analytical results reflects the accuracy of our proposed mathematical model and its effectiveness in capturing the underlying SER of the system. It is observed that the performance of the system is improved by increasing the value of $\alpha$, i.e., as the noise goes from Gamma to Laplacian and to Gaussian, as discussed in Section \ref{model}. Nonetheless, it is shown that as $\alpha$ increases, its effect on the error rate performance decreases. This finding is in agreement with the results demonstrated in \cite{9097184, soury2013}. Furthermore, it is evident that increasing the number of RIS elements causes a significant reduction in the SER, due to the enhanced spatial diversity gain offered by RIS-assisted systems. This make them particularly attractive for large-scale applications. Also, by closely inspecting Fig. \ref{fig:Pe_SNR}, it is noticed that adding more RIS elements is more rewarding in the high SNR regime than the low SNR regime, where the reduction in the SER is more significant. It can be further observed from Fig. \ref{fig:Pe_SNR} that, although the effect of the noise type is more pronounced when $N$ is large, the overall SER behavior is the same for all $N$ values. We emphasize that the analytical expressions derived in this work lend themselves as effective tools in identifying the minimum number of RIS elements that should be deployed in order to achieve a targeted SER value under different noise types. As a consequence, unnecessary phase adjustments to account for the extra deployed RIS elements may be avoided.
\par Fig. \ref{fig:d_SNR} depicts the achievable diversity order versus the average transmit SNR when Laplacian and Gaussian noise are considered for the cases when $N=5,10$ and 20. It can be noted from the figure that the diversity order converges to $(a_5+1)/2$ regardless of the type of noise that corrupts the received signals. Interestingly, it is observed that the diversity order is more impacted by the variation of the value of $\alpha$ at low SNR compared to high SNR scenarios. It can be further noticed that for a high number of RIS elements, the diversity order of the system is highly affected by $\alpha$ compared to a low number of RE. It is worth highlighting that a similar conclusion can be obtained in a multi-user scenario, where it was demonstrated in \cite{9097184} that $\alpha$ has no impact on the users' diversity orders. This implies that in an RIS-enabled multi-user scenario, e.g., NOMA, the number of REs and the users' order determine the achieved diversity order.

\par Fig. \ref{fig:distance} demonstrate the impact of the RIS location on the error rate performance of the underlying system model, in the presence of different noise types, namely, Laplacian and Gaussian noise. Here, we set the total distance between $S$ and $D$ to 5~m, i.e., $d_1 + d_2 = 5$, and $\varrho=2.7$. From the figure, it can be observed that the best error rate performance is obtained when the RIS is closer to the transmitter or the receiver, while the worse performance is experienced when the RIS is located in the middle between $S$ and $D$. While it is shown that different noise types over different RIS locations yield similar behaviors, the effect of the noise type with respect to the RIS location is more pronounced for higher $N$ values.

\begin{figure}[!t]
\centering
\includegraphics[width=0.8\linewidth]{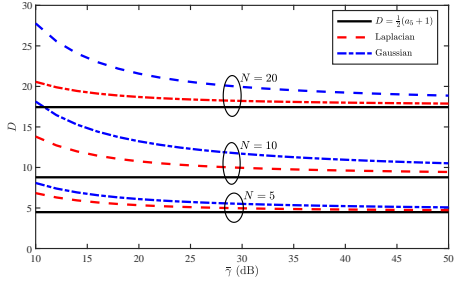}
\caption{Achievable diversity order versus $\bar{\gamma}$ for $N=5$, 10, and 20, in the presence of Laplacian and Gaussian noise.}
\vspace*{-0.5cm}
\label{fig:d_SNR}
\end{figure}

\begin{figure}[!t]
\centering
\includegraphics[width=0.8\linewidth]{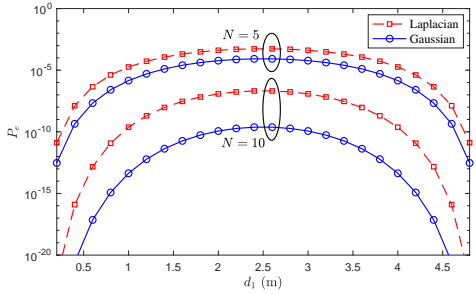}
\caption{Average SER versus $d_1$ for Laplacian and Gaussian noise scenarios, $\bar{\gamma} = 20$ dB, $N=5,10$.}
\vspace*{-0.5cm}
\label{fig:distance}
\end{figure}

\section{Conclusions}
\label{sec:conc}

In this letter, we examined the performance of RIS-assisted systems subject to GGN and considering Rayleigh fading. In particular, we derived a novel closed-form expression for the SER in the considered setup and studied the impact of the GGN parameters on the SER performance and achievable diversity order. Our results demonstrated that the noise shaping parameter has a negligible impact on the diversity order of the system and that the maximum achievable diversity order is identified by the number of RIS elements, $N$. Our proposed framework can be considered as the baseline for many future extensions and research directions, such as in multiple users/RIS/cells, multiple access-based systems, e.g., NOMA, performance under hardware impairment, and RIS-aided unmanned aerial vehicles networks, to name a few.

\vspace{-2pt}

\bibliographystyle{IEEEtran}
\bibliography{references}
\balance

\end{document}